\pdfoutput=1
\RequirePackage{ifpdf}
\ifpdf 
\documentclass[pdftex]{sigma}
\else
\documentclass{sigma}
\fi

\usepackage{amscd}
\usepackage[all]{xy}
\usepackage{cancel}

\numberwithin{equation}{section}

\newtheorem{Theorem}{Theorem}[section]
\newtheorem{Corollary}[Theorem]{Corollary}
\newtheorem{Proposition}[Theorem]{Proposition}
 { \theoremstyle{definition}
\newtheorem{Definition}[Theorem]{Definition}
\newtheorem{Remark}[Theorem]{Remark} }

\begin{document}

\allowdisplaybreaks

\renewcommand{\thefootnote}{$\star$}

\newcommand{\arXivNumber}{1511.00234}

\renewcommand{\PaperNumber}{023}

\FirstPageHeading

\ShortArticleName{Haantjes Structures for the Jacobi--Calogero Model and the Benenti Systems}

\ArticleName{Haantjes Structures for the Jacobi--Calogero Model\\ and the Benenti Systems\footnote{This paper is a~contribution to the Special Issue
on Analytical Mechanics and Dif\/ferential Geometry in honour of Sergio Benenti.
The full collection is available at \href{http://www.emis.de/journals/SIGMA/Benenti.html}{http://www.emis.de/journals/SIGMA/Benenti.html}}}

\Author{Giorgio TONDO~$^\dag$ and Piergiulio TEMPESTA~$^{\ddag\S}$}

\AuthorNameForHeading{G.~Tondo and P.~Tempesta}

\Address{$^\dag$~Dipartimento di Matematica e Geoscienze, Universit\`a degli Studi di Trieste,\\
\hphantom{$^\dag$}~piaz.le Europa 1, I--34127 Trieste, Italy}
\EmailD{\href{mailto:tondo@units.it}{tondo@units.it}}

\Address{$^\ddag$~Departamento de F\'{\i}sica Te\'{o}rica II, Facultad de F\'{\i}sicas, Universidad Complutense,\\
\hphantom{$^\ddag$}~28040 -- Madrid, Spain}

\Address{$^\S$~Instituto de Ciencias Matem\'aticas (CSIC-UAM-UC3M-UCM), C/ Nicol\'as Cabrera,\\
\hphantom{$^S$}~No 13--15, 28049 Madrid, Spain}
\EmailD{\href{mailto:ptempest@ucm.es}{ptempest@ucm.es}}
\URLaddressD{\url{http://www.icmat.es/p.tempesta}}

\ArticleDates{Received November 03, 2015, in f\/inal form February 22, 2016; Published online March 03, 2016}

\Abstract{In the context of the theory of symplectic-Haantjes manifolds, we construct the Haantjes structures of generalized St\"{a}ckel systems and, as a particular case, of the quasi-bi-Hamiltonian systems. As an application, we recover the Haantjes manifolds for the rational Calogero model with three particles and for the Benenti systems.}


\Classification{37J35; 70H06; 70H20; 53D05}

\rightline{\textit{To Sergio Benenti, on the occasion of his 70th birthday.}}

\renewcommand{\thefootnote}{\arabic{footnote}}
\setcounter{footnote}{0}

\section{Introduction}

The purpose of this paper is to present an application of the theory recently developed in \cite{TT2014}, aiming at a characterization of integrability and separability for classical Hamiltonian systems by means of the geometry of Haantjes operators.

In \cite{TT2014} the notion of Haantjes manifolds has been introduced in the realm of integrability for f\/inite-dimensional systems (see also \cite{FeMa,MFrob} and \cite{MGall} for a treatment of integrable hierarchies of PDEs from a dif\/ferent perspective). Haantjes tensors \cite{Haa} represent a natural generalization of the well known Nijenhuis tensors \cite{Nij,Nij2,Nij2+}; the $(1,1)$ tensor f\/ields with vanishing Haantjes tensor encode many crucial features of an integrable system, especially in relation with the property of separability of the associated Hamilton--Jacobi (H-J) equation. Due to the fact that these new geometrical structures are also endowed with a standard symplectic structure, we call them symplectic-Haantjes or $\omega \mathcal{H}$ \textit{manifolds}.

The theory of Haantjes manifolds is very general: it encompasses essentially all known results concerning integrability of f\/inite classical systems.

A particularly relevant class of integrable models are the \textit{separable} ones: one can f\/ind a coordinate system in which the H-J equation takes a~separated form. In this f\/ield, the contribution of Benenti has been crucial. One of his theorems~\cite{Ben80}, particularly useful for us, states that a~family of Hamiltonian functions $\{H_i\}_{1\le i\le n}$ is separable in a set of
canonical coordinates $(\boldsymbol{q};\boldsymbol{p})$ if and only if they
are in {\rm separable involution}, that is, they satisfy the conditions
\begin{gather*} \label{eq:SI}
\{H_i,H_j\}_{\vert k}=\frac{\partial H_i}{\partial q_k}
\frac{\partial H_j}{\partial p_k}-\frac{\partial H_i}{\partial p_k}
\frac{\partial H_j}{\partial q_k}=0 , \qquad 1\le k \le n,
\end{gather*}
where no summation over $k$ is understood.

The problem of separation of variables (SoV) can be recast and treated in our approach. Indeed, under mild hypotheses, from the Haantjes structure associated with an integrable system one can derive a set of coordinates, that we shall call the \textit{Darboux--Haantjes coordinates}, representing separation coordinates for the system.

The Haantjes structure associated with an integrable system allows a tensorial description (i.e., intrinsic) of its main properties.
		
In this paper, we shall prove that some of the most relevant separable systems, namely the generalized St\"ackel systems, are described in terms of suitable Haantjes structures, which are responsible of their separation properties. On one hand, by identifying the St\"ackel matrix in their def\/inition with a Vandermonde type matrix, one can obtain the class of quasi-bi-Hamiltonian systems. On the other hand, the standard St\"ackel systems are re-obtained by specializing properly the St\"ackel functions.

Another important aspect of our analysis is that two particularly relevant integrable systems, namely the Jacobi--Calogero model and the family of Benenti systems, can be studied in the framework of Haantjes geometry.

The paper is organized as follows. In Section~\ref{section2}, we review the main properties of Haantjes operators and Lenard--Haantjes chains, together with the problem of separation of variables in the context of Haantjes geometry. In Section~\ref{section3}, we construct the Haantjes structure of the gene\-ra\-lized St\"ackel systems. As a particular case, the Haantjes structure for the quasi-bi-Hamiltonian systems and for a Goldf\/ish model~\cite{CalIso} is obtained in Section~\ref{section4}. In Section~\ref{section5}, the classical St\"ackel systems and their Killing tensors are discussed. In Section~\ref{section6}, three Haantjes structures for the Jacobi--Calogero model are presented. In the f\/inal Section~\ref{section7}, Haantjes manifolds for the family of Benenti systems are obtained.

\section{The Haantjes geometry: a brief review}\label{section2}

We shall f\/irst summarize some basic results of the theory of Nijenhuis and Haantjes tensors, following \cite{Haa,Nij} (see also \cite{FN,Nij2,Nij2+}).

\subsection{Nijenhuis and Haantjes operators}
\label{sec:1}

Let $M$ be a dif\/ferentiable manifold and $\boldsymbol{L}\colon TM\rightarrow TM$ be a~(1,1) tensor f\/ield (that is, a f\/ield of linear operators on the tangent space at each point of~$M$).

\begin{Definition}\looseness=-1
The
 \textit{Nijenhuis torsion} of $\boldsymbol{L}$ is the skew-symmetric $(1,2)$ tensor f\/ield def\/ined by
\begin{gather*} 
\mathcal{T}_ {\boldsymbol{L}} (X,Y):=\boldsymbol{L}^2[X,Y] +[\boldsymbol{L}X,\boldsymbol{L}Y]-\boldsymbol{L}\big([X,\boldsymbol{L}Y]+[\boldsymbol{L}X,Y]\big),
\end{gather*}
where $X,Y \in TM$ and $[ \,,\, ]$ denotes the commutator of two vector f\/ields.
\end{Definition}

Given a set of local coordinates $\boldsymbol{x}=(x_1,\ldots, x_n)$ in $M$, the Nijenhuis torsion takes the form
\begin{gather*}
(\mathcal{T}_{\boldsymbol{L}})^i_{jk}=\sum_{\alpha =1}^n\left( \frac{\partial {\boldsymbol{L}}^i_k} {\partial x_\alpha} {\boldsymbol{L}}^\alpha_j -\frac{\partial {\boldsymbol{L}}^i_j} {\partial x_\alpha} {\boldsymbol{L}}^\alpha_k+\left(\frac{\partial {\boldsymbol{L}}^\alpha_j} {\partial x_k} -\frac{\partial {\boldsymbol{L}}^\alpha_k} {\partial x_j}\right) {\boldsymbol{L}}^i_\alpha \right),
\end{gather*}
with $n^2(n-1)/2$ independent components.

\begin{Definition}
The \textit{Haantjes tensor} associated with $\boldsymbol{L}$ is the (1,2) tensor f\/ield def\/ined by
\begin{gather*} 
\mathcal{H}_{\boldsymbol{L}}(X,Y):=\boldsymbol{L}^2\mathcal{T}_{\boldsymbol{L}}(X,Y)
+\mathcal{T}_{\boldsymbol{L}}(\boldsymbol{L}X,\boldsymbol{L}Y)
-\boldsymbol{L}\big(\mathcal{T}_{\boldsymbol{L}}(X,\boldsymbol{L}Y)+\mathcal{T}_{\boldsymbol{L}}(\boldsymbol{L}X,Y)\big).
\end{gather*}
\end{Definition}

The Haantjes tensor is skew-symmetric due to the skew-symmetry of the Nijenhuis torsion. Locally, we have
\begin{gather*}
(\mathcal{H}_{\boldsymbol{L}})^i_{jk}= \sum_{\alpha,\beta =1}^n\big(\boldsymbol{L}^i_\alpha \boldsymbol{L}^\alpha_\beta(\mathcal{T}_{\boldsymbol{L}})^\beta_{jk} +
(\mathcal{T}_{\boldsymbol{L}})^i_{\alpha \beta}\boldsymbol{L}^\alpha_j \boldsymbol{L}^\beta_k-
\boldsymbol{L}^i_\alpha\big( (\mathcal{T}_{\boldsymbol{L}})^\alpha_{\beta k} \boldsymbol{L}^\beta_j+
 (\mathcal{T}_{\boldsymbol{L}})^\alpha_{j \beta } \boldsymbol{L}^\beta_k \big)\big).
\end{gather*}

A discussion of some basic examples can be found in~\cite{TT2014}.

The powers of a single Haantjes operator generate a module over the ring of smooth functions on $M$, as can be inferred from the results in~\cite{BogCMP,BogI}.
Our main def\/initions are the following.
\begin{Definition}
A Haantjes (Nijenhuis) f\/ield of operators is a f\/ield of operators whose associated Haantjes (Nijenhuis) tensor vanishes identically.
\end{Definition}

\begin{Definition}
 A f\/ield of operators $\boldsymbol{L}$ is said to be \textit{semisimple} if is diagonalizable at each point~$\boldsymbol{x}$ of~$M$.
\end{Definition}

The following important result, proved by Haantjes, establishes the conditions assuring that the generalized eigen-distribution $\mathcal{D}_i :=\operatorname{Ker}(\boldsymbol{L}-l_{i}\boldsymbol{I})^{\rho_i}$ of $\boldsymbol{L}$ ($\rho_i$ denotes the Riesz index of the eigenvalue $l_i$) is integrable.

\begin{Theorem}[\cite{Haa}] \label{th:Haan}
Let $\boldsymbol{L}$ be a generic field of operators, and assume that the rank of each generalized eigen-distribution $\mathcal{D}_i$ is independent of $\boldsymbol{x}\in M$. The vanishing of the Haantjes tensor
\begin{gather} \label{eq:HaaNullTM}
\mathcal{H}_{\boldsymbol{L}}(X, X')= 0, \qquad \forall\, X, X' \in TM
\end{gather}
is a sufficient condition to ensure the
integrability of each distribution $\mathcal{D}_i$
and of any direct sum $\mathcal{D}_i \oplus \mathcal{D}_j \oplus \cdots \oplus \mathcal{D}_k$ $($where all indices $i,j,\ldots, k$ are different$)$. In addition, if~$\boldsymbol{L}$ is semisimple, the converse is also true.
\end{Theorem}

We remind that a reference frame is a set of~$n$ vector f\/ields $\{Y_1,\ldots,Y_n\}$
 which form a~basis of the tangent space $T_{\boldsymbol{x}}U$ at each point~$\boldsymbol{x}$ belonging to an open set $U\subseteq M$.
 Two frames $\{X_1,\ldots,X_n\}$ and $\{Y_1,\ldots,Y_n\}$ are said to be equivalent if $n$ nowhere vanishing smooth functions $f_i$ do exist such that
 \begin{gather*}
 X_i= f_i(\boldsymbol{x}) Y_i , \qquad i=1,\ldots,n .
 \end{gather*}
 A natural or coordinate frame $\big\{\frac{\partial}{\partial x_1}, \ldots, \frac{\partial}{\partial x_n}\big\}$ is the frame associated to a local chart $\{(x_1,\ldots$, $x_n)\}$.
 \begin{Definition}
 An \emph{integrable} frame is a reference frame equivalent to a natural frame.
 \end{Definition}

In other words, to say that a frame $\{Y_1,\ldots,Y_n\}$ is integrable there must exist a local chart $(x_1, \ldots, x_n)$ and $n$ nowhere vanishing functions $f_i$ such that
\begin{gather*}
 Y_i =f_i(\boldsymbol{x}) \frac{\partial}{\partial x_i} ,\qquad i=1,\ldots , n .
\end{gather*}

\begin{Proposition}[\cite{BCRframe}] \label{pr:BCRframe}
A reference frame in a differentiable manifold~$M$ is an integrable frame if and only if it satisfies any of the two equivalents conditions:
\begin{itemize}\itemsep=0pt
\item
each two-dimensional distribution generated by any two vector fields $Y_i$, $Y_j$ is Frobenius integrable;
\item
each $(n-1)$-dimensional distribution $E_i$ generated by all the vector fields except~$Y_i$ is Frobenius integrable.
\end{itemize}
\end{Proposition}

Then, under the hypotheses of Theorem~\ref{th:Haan}, we can interpret equation~\eqref{eq:HaaNullTM} as the suf\/f\/icient condition that ensures the existence of a suitable integrable generalized eigen-frame of~$\boldsymbol{L}$. Furthermore, if $\boldsymbol{L}$ is semisimple, condition~\eqref{eq:HaaNullTM} is also necessary.

\subsection{Symplectic-Haantjes manifolds: an outline}
The symplectic-Haantjes manifolds, or $\omega \mathcal{H}$ manifolds, have been introduced in \cite{TT2014}. As we shall see in the subsequent sections, these structures allow to formulate the theory of Hamiltonian integrable systems in a natural geometric language.
\begin{Definition}\label{def:HM}
A symplectic-Haantjes or $\omega \mathcal{H}$ manifold $(M, \omega, \boldsymbol{K}_0, \boldsymbol{K}_1,\ldots,\boldsymbol{K}_{n-1})$ is a symplectic manifold of dimension $2n$, endowed with $n$ endomorphisms of the tangent bundle of~$M$
 \begin{gather*}
 \boldsymbol{K}_\alpha\colon \ TM\mapsto TM ,\qquad \alpha=0,\ldots , n-1,
 \end{gather*}
 which satisfy the following conditions:
\begin{itemize}\itemsep=0pt
 \item
The operator $\boldsymbol{K}_0$ is the identity operator in $TM$
\begin{gather*} 
\boldsymbol{K}_0=\boldsymbol{I} .
\end{gather*}
\item
Their Haantjes tensor vanishes identically, that is
\begin{gather*} 
\mathcal{H}_{\boldsymbol{K}_\alpha}(X,Y)=0, \qquad \forall \, X,Y \in TM, \qquad \alpha=0,\ldots, n-1.
\end{gather*}
\item
The endomorphisms are compatible with $\omega$ (or equivalently, with the corresponding symplectic operator
 $\boldsymbol{\Omega}:=\omega ^\flat$), namely
\begin{gather} \label{eq:LOcomp}
\boldsymbol{K}_\alpha^T\boldsymbol{\Omega}=\boldsymbol{\Omega} \boldsymbol{K}_\alpha, \qquad \alpha=0,\ldots , n-1 ,
\end{gather}
where $\boldsymbol{K}_\alpha^T\colon T^{*}M\mapsto T^{*}M$ is the transposed operator of $\boldsymbol{K}_\alpha$.
\item
The endomorphisms are compatible with each others, in the sense that they form a commutative ring
\begin{gather} \label{eq:Lcomp}
\boldsymbol{K}_{\alpha}\boldsymbol{K}_\beta=\boldsymbol{K}_\beta \boldsymbol{K}_\alpha, \qquad \alpha,\beta=0,\ldots , n-1 ,
\end{gather}
and generate a module $\mathcal{K}$ over the ring of smooth functions on $M$, that is, they satisfy
\begin{gather} \label{eq:LtildeC}
\mathcal{H}_{\big(\sum\limits_{\alpha=0}^{n-1} a_{\alpha}(\boldsymbol{x})\boldsymbol{K}_{\alpha} \big)}(X,Y)= 0
 , \qquad\forall \, X, Y \in TM ,
\end{gather}
where $a_{\alpha}(\boldsymbol{x})$ are arbitrary smooth functions on M.
\end{itemize}
The $(n+1)$-tuple $(\omega, \boldsymbol{K}_0, \boldsymbol{K}_1,\ldots,\boldsymbol{K}_{n-1})$ is called the $\omega \mathcal{H}$ structure associated with the $\omega \mathcal{H}$ manifold, and the module (ring)
$\mathcal{K}$ is called the Haantjes module (ring).
 \end{Definition}

In other words, we require that the endomorphisms $\boldsymbol{K}_\alpha$ and any operator belonging to the module
(ring) $\mathcal{K}$ be a Haantjes operator compatible with $\omega$ and with the original Haantjes operators $\{\boldsymbol{K}_0, \boldsymbol{K}_1,\ldots,\boldsymbol{K}_{n-1}\}$.

\subsection{Lenard--Haantjes chains}
Despite the relevance of Lenard chains in soliton hierarchies, especially in the construction of integrals of motion in involution \cite{Magri78,MagriLE, MM}, their importance in the theory of separation of variables for f\/inite-dimensional Hamiltonian systems has been acknowledged only recently (see \cite{FMT,FP,MLenard,MT, MTlt,MTltC,TT,TKdV,TGalli12}).
The natural extension of the original notion of Lenard chain to the context of the Haantjes geometry is proposed below.
\begin{Definition}\label{defi:LHc}
 Let $(M, \omega, \boldsymbol{K}_0, \boldsymbol{K}_1,\ldots,\boldsymbol{K}_{n-1})$ be a $2n$-dimensional
$\omega \mathcal{H}$ manifold and let\linebreak
$\{H_j\}_{1\le j \le n}$ be $n$ independent functions which satisfy the following relations
\begin{gather*} \label{eq:LHc}
\mathrm{d} H_j=\boldsymbol{K}_\alpha^T \mathrm{d} H, \qquad j=\alpha +1, \qquad \alpha=0,\ldots, n-1, \qquad H:=H_1.
\end{gather*}
Under these conditions, we shall say that the functions $\{H_j\}_{1\le j \le n}$ form a Lenard--Haantjes chain generated by the function~$H$.
\end{Definition}

The relevance of Lenard--Haantjes chains is clarif\/ied by the following

\begin{Proposition}\label{pr:Hinv}
Let $M$ be a $2n$-dimensional $\omega \mathcal{H}$ manifold and $\{H_j\}_{1\le j \le n}$ be $n$ smooth independent functions forming a Lenard--Haantjes chain. Then, the foliation generated by these functions is Lagrangian. Consequently, each Hamiltonian system with Hamiltonian func\-tions~$H_j$, $1\le j \le n$ is integrable by quadratures.
\end{Proposition}
\begin{proof}
 By virtue of the classical Arnold--Liouville theorem, it is suf\/f\/icient to prove that the functions $H_j$ belonging to a Lenard--Haantjes chain are in involution w.r.t.\ the Poisson bracket def\/ined by the symplectic form~$\omega$. In fact,
 if we denote by $\boldsymbol{P}:=\boldsymbol{\Omega}^{-1}$ the Poisson operator induced by the symplectic form, we obtain
\begin{gather*}
\{H_j, H_k\}=\langle dH_j,\boldsymbol{P} dH_k\rangle =\big\langle \boldsymbol{K}_{\alpha}^T dH,\boldsymbol{P} \boldsymbol{K}^T_{\beta} dH\big\rangle =
\big\langle dH,\boldsymbol{K}_{\alpha} \boldsymbol{P} \boldsymbol{K}_{\beta}^T dH\big\rangle =0 ,
\end{gather*}
as the operator $\boldsymbol{K}_{\alpha} \boldsymbol{P} \boldsymbol{K}_{\beta}^T$ is skew-symmetric by virtue of the compatibility condition \eqref{eq:LOcomp}.
\end{proof}

\begin{Remark} \label{rem:nK}
Given a $\omega \mathcal{H}$ manifold, the purpose of its Haantjes operators is to provide a~Lenard--Haantjes chain of~$n$ integrals of motion in involution. To this end, $n$ independent Haantjes operators are required.
\end{Remark}

\subsection[Darboux--Haantjes coordinates and separation of variables in $\omega \mathcal{H}$ manifolds]{Darboux--Haantjes coordinates and separation of variables\\ in $\boldsymbol{\omega \mathcal{H}}$ manifolds}

By analogy with the classical Darboux coordinates, in Haantjes geometry the Darboux--Haantjes coordinates (DH) are a set of distinguished local symplectic coordinates, which simultaneously diagonalize every Haantjes operator.
\begin{Definition}
Let $(M, \omega, \boldsymbol{K}_0, \boldsymbol{K}_1,\ldots,\boldsymbol{K}_{n-1})$ be a $\omega \mathcal{H}$ manifold. A set of local coordinates $(q_1,\ldots,q_n ; p_1,\ldots,p_n)$ will be said to be a set of Darboux--Haantjes (DH) coordinates if the symplectic form in these coordinates assumes the \emph{Darboux form}
\begin{gather*}
\omega=\sum _{i=1}^n \mathrm{d} p_i \wedge \mathrm{d} q_i
 \end{gather*}
and each Haantjes operator diagonalizes:
\begin{gather*}
\boldsymbol{K}_\alpha=\sum _{i=1}^n l_{i } ^{(\alpha)}( \boldsymbol{q},\boldsymbol{p})\left(\frac{\partial}{\partial q_i}\otimes \mathrm{d} q_i + \frac{\partial}{\partial p_i}\otimes \mathrm{d} p_i\right ), \qquad \alpha=0,\ldots,n-1,
 \end{gather*}
with $l_{i}^{(0)}=1$, $i=1,\ldots,n$.
\end{Definition}

In~\cite{TT2014} we have shown that a semisimple Haantjes structure which admits a maximal generator, that is a cyclic Haantjes operator with~$n$ distinct eigenvalues, provides DH coordinates. They turn out to be \textit{separation variables} for each Hamiltonian function belonging to an asso\-ciated Lenard--Haantjes chain. In particular, see Theorems~57 and~59 of~\cite{TT2014} for the statement and the proof of the main results concerning the existence of separation variables for~$\omega \mathcal{H}$ manifolds. By virtue of these general theorems, such structures take the form
 \begin{gather} \label{eq:LSoV}
\boldsymbol{K}_\alpha=\sum _{i=1}^n \frac{\frac{\partial H_{\alpha +1}}{\partial p_i}}{ \frac{\partial H_1}{\partial p_i}}\left(\frac{\partial}{\partial q_i}\otimes \mathrm{d} q_i +\frac{\partial}{\partial p_i}\otimes \mathrm{d} p_i \right ), \qquad \alpha=0,\ldots,n-1 ,
\end{gather}
in any set of separable Darboux coordinates $(\boldsymbol{q},\boldsymbol{p})$ for the Hamiltonian functions~$H_\alpha$, $\alpha=0,\ldots,n-1$.

\section{Generalized St\"ackel systems}\label{section3}

The purpose of this section is to determine the Haantjes structure for a huge class of Hamiltonian systems of St\"ackel type, that we shall call the generalized St\"ackel systems. Specializing conveniently their St\"ackel matrix and functions, we can also derive, in a direct way, both families of quasi-bi-Hamiltonian and classical St\"ackel systems.
\begin{Proposition}[generalized St\"ackel systems] \label{Stackelprop}
 Let us consider the Hamiltonian functions~{\rm \cite{AKN}}
\begin{gather} \label{eq:HStack}
H_j=\sum_{k=1}^n\frac{\tilde{S}_{jk}}{\det (\boldsymbol{S})} f_k(q_k,p_k), \qquad j=1,\dots,n ,
 \end{gather}
where $S_{ij}$ are the elements of a St\"ackel matrix $\boldsymbol{S}(\boldsymbol{q})$ $($i.e., an invertible matrix whose i-th row depends on the coordinate $q_i$ only$)$, and $\tilde{S}_{jk}$ denotes the cofactor of the element~$S_{kj}$.
They belong to the Lenard--Haantjes chain
\begin{gather} \label{eq:HcS}
\boldsymbol{K}_{j-1} ^T\mathrm{d} H_1=\mathrm{d} H_{j},\qquad j=1,\ldots,n ,
\end{gather}
where $\boldsymbol{K}_{j-1}$ are the Haantjes operators defined by
\begin{gather} \label{eq:StackHaan}
\boldsymbol{K}_{j-1}:=\sum_{r=1}^n \frac{ \tilde{S}_{jr}}{ \tilde{S}_{1r}} \left( \frac{\partial}{\partial q_r}\otimes \mathrm{d} q_r+ \frac{\partial}{\partial p_r}\otimes \mathrm{d} p_r\right),
 \qquad j=1,\ldots,n .
\end{gather}
\end{Proposition}

\begin{proof}
The operators \eqref{eq:StackHaan} are a special case of the operators~\eqref{eq:LSoV} for the St\"ackel Hamiltonian functions~\eqref{eq:HStack}.
\end{proof}

\begin{Remark}
It is known that the St\"ackel matrix $\boldsymbol{S}(\boldsymbol{q})$ is not unique. In fact, multiplying the $i$-th row of a given St\"ackel matrix for an arbitrary function $F_i(q_i)$, one obtains a~dif\/ferent St\"ackel matrix for the same coordinate web. However, it should be noted that, although the Hamiltonian functions~\eqref{eq:HStack} transform into
\begin{gather*}
H_j \mapsto \sum_{k=1}^n\frac{\tilde{S}_{jk}}{F_k(q_k) \det (\boldsymbol{S})} f_k(q_k,p_k), \qquad j=1,\dots,n ,
\end{gather*}
 the Haantjes operators~\eqref{eq:StackHaan} stay \emph{invariant}, as their eigenvalues turn out to be
\begin{gather*}
\frac{ \tilde{S}_{jr}}{\cancel{F_r(q_r) }}\frac{\cancel{ F_r(q_r)}}{ \tilde{S}_{1r}} \equiv \frac{ \tilde{S}_{jr}}{ \tilde{S}_{1r}} .
\end{gather*}
For this reason, we could say that the Haantjes operators \eqref{eq:StackHaan} are the tensorial representation of a given St\"ackel web in~$T^*\mathcal{Q}$.
\end{Remark}

\begin{Remark}
The Haantjes operators~\eqref{eq:StackHaan} are independent of the functions $f_k(q_k,p_k)$, that appear in the Hamiltonians~\eqref{eq:HStack} only. They are called the \emph{St\"ackel} functions and are characteristic functions of the Haantjes web, according to~\cite{TT2014}.
\end{Remark}

By choosing as St\"ackel functions $f_k=\psi_k(p_k)$, $ f_k=W_k(q_k)$ and $f_k=\psi_k(p_k)+W_k(q_k)$, where $\psi_k(p_k)$ and $ W_k(q_k)$ are arbitrary smooth functions of their argument, from equation~\eqref{eq:HcS} and Proposition~\ref{pr:Hinv}, we obtain the following result.
\begin{Corollary}
The functions
\begin{gather} \label{eq:StackTV}
T_j:=\sum_{k=1}^n\frac{ \tilde{S}_{jk}}{\det (\boldsymbol{S})} \psi_k(p_k) , \qquad
 V_j:=\sum_{k=1}^n\frac{ \tilde{S}_{jk}}{\det (S)} W_k(q_k)
 \qquad j=1,\dots,n ,
\end{gather}
are elements of the Haantjes chains
\begin{gather} \label{eq:HcTV}
\boldsymbol{K}_{j-1}^T \mathrm{d} T_1=\mathrm{d} T_{j} , \qquad \boldsymbol{K}_{j-1} ^T\mathrm{d} V_1=\mathrm{d} V_{j} ,\qquad j=1,\ldots,n .
\end{gather}
Therefore, they fulfill the involution relations
\begin{gather} \label{eq:invTV}
\{ T_i,T_j\}=0 , \qquad \{V_i,V_j\}=0 , \qquad\{T_i,V_j\}+\{V_i,T_j\}=0 , \qquad i,j=1,\ldots, n .
\end{gather}
\end{Corollary}

Using~\cite[Theorem~57]{TT2014}, stating the existence of generators of a~Haantjes structure, we can derive the following result.
\begin{Proposition}
The Haantjes structure of a $($generalized$)$ St\"ackel system admits as gene\-ra\-tor the Haantjes operator
\begin{gather} \label{eq:KStack}
\boldsymbol{K}:=\sum_{r=1}^n \lambda_r(\boldsymbol{q})\left( \frac{\partial}{\partial q_r}\otimes \mathrm{d} q_r+ \frac{\partial}{\partial p_r}\otimes \mathrm{d} p_r\right) ,
 \end{gather}
 where $\lambda_r(\boldsymbol{q})$ are arbitrary smooth functions of the coordinates $(q_1,\ldots,q_n)$, with $\lambda_r\neq \lambda_j$ at any point of~$\mathcal{Q}$, except
possibly for a closed set.

In fact,
\begin{gather*}\label{eq:Kgen}
 \boldsymbol{K}_{j-1}=\sum _{r=1}^n \frac{ \tilde{S}_{jr}}{ \tilde{S}_{1r}}
 \pi_r(\boldsymbol{K}),
 \qquad j=1,\ldots,n ,
\end{gather*}
where the operators
\begin{gather*}
\pi_r(\boldsymbol{K}):= \frac{\prod\limits_{\stackrel{i=1}{i\neq r}}^n (\boldsymbol{K}-\lambda_i\boldsymbol{I})}
{\prod\limits_{\stackrel{i=1}{i\neq r}}^n (\lambda_r-\lambda_i)}
= \frac{\partial}{\partial q_r}\otimes \mathrm{d} q_r+ \frac{\partial}{\partial p_r}\otimes \mathrm{d} p_r, \qquad r=1,\ldots,n ,
\end{gather*}
are the elements of the so-called Lagrange interpolation basis $\mathcal{B}_{\rm int}=\{ \pi_1(\boldsymbol{K}), \ldots,\pi_n(\boldsymbol{K})\}$ associated to the operator $\boldsymbol{K}$, and represent a basis of the Haantjes module $\mathcal{K}$.
\end{Proposition}
\begin{Remark}
The representation of the Haantjes operators~\eqref{eq:StackHaan} on the cyclic basis $\mathcal{B}_{\rm cycl}=\{ \boldsymbol{I},\boldsymbol{K},\boldsymbol{K}^2,\ldots, \boldsymbol{K}^{n-1}\}$ associated to the Haantjes operator \eqref{eq:KStack} can be obtained by observing that the transition matrix between the cyclic basis and the interpolation basis is given by the Vandermonde matrix of the eigenvalues of \eqref{eq:KStack}
\begin{gather*}
[\boldsymbol{I} ]_{\mathcal{B}_{\rm cycl}} ^{ \mathcal{B}_{\rm int}} =\boldsymbol{V}(\boldsymbol{q})=\begin{bmatrix}
 1 & \lambda_1^{2} &\ldots&\lambda_1^{n-1} \\
 1& \lambda_2^{2} &\ldots& \lambda_2^{n-1} \\
 \ldots & \ldots & \ldots & \ldots \\
 1 & \lambda_n^{2} &\ldots& \lambda_n^{n-1}
\end{bmatrix} .
\end{gather*}
Thus, the Haantjes operators \eqref{eq:StackHaan} can be written as polynomial f\/ields in the powers of $\boldsymbol{K}$
\begin{gather} \label{eq:pK}
\boldsymbol{K}_{j-1}=p_{j-1} (\boldsymbol{q},\boldsymbol{K})=\sum_{i =0} ^{n-1} a_{i+1}^{(j-1)}(\boldsymbol{q}) \boldsymbol{K}^i, \qquad j=1,\ldots,n ,
\end{gather}
where
\begin{gather*}
\begin{bmatrix}
 a_1^{(j-1)} \\
 a_2^{(j-1)} \\
 \ldots \\
 a_n^{(j-1)} \
 \end{bmatrix}
 ^{\mathcal{B}_{\rm cycl}}
= \boldsymbol{V}^{-1}
\begin{bmatrix}
 \frac{ \tilde{S}_{j1}}{ \tilde{S}_{11}} \\
 \frac{ \tilde{S}_{j2}}{ \tilde{S}_{12} } \\
 \ldots \\
 \frac{ \tilde{S}_{jn}}{ \tilde{S}_{1n} }
\end{bmatrix}
^{\mathcal{B}_{\rm int}}
 .
\end{gather*}
Another basis of interest for the sequel is the so-called \emph{control basis} $\mathcal{B}_{\rm cont}=\{ e_1(\boldsymbol{K}), \ldots,e_n(\boldsymbol{K})\}$ associated with the operator $\boldsymbol{K}$ (see, for instance, \cite[p.~98]{Fur}). Its elements (in reverse order) are def\/ined by
\begin{gather}
e_1(\boldsymbol{K}) = \boldsymbol{I},\nonumber \\
e_2(\boldsymbol{K}) = - c_{1}\boldsymbol{I}+\boldsymbol{K}, \nonumber\\
 \cdots \cdots\cdots\cdots\cdots\cdots\cdots\cdots\cdots \nonumber\\
 e_n(\boldsymbol{K}) = - c_{n-1}\boldsymbol{I}-c_{n-2}\boldsymbol{K}-\cdots-c_1\boldsymbol{K}^{n-2}
+\boldsymbol{K}^{n-1},\label{eq:ContBase}
\end{gather}
where the functions $c_1(\boldsymbol{q}), \ldots, c_n(\boldsymbol{q})$ are the (opposite of the) coef\/f\/icients of the minimal polynomial of $\boldsymbol{K}$
\begin{gather} \label{eq:mN}
m_{\boldsymbol{K}} (\lambda) =\lambda^n-c_1 \lambda^{n-1}-\cdots -c_{n-1}\lambda-c_n .
\end{gather}
Thus, these coef\/f\/icients are related to the elementary symmetric functions $\sigma_k$ of the roots of~\eqref{eq:mN}, namely the $n$ eigenvalues $( \lambda_1, \lambda_2,\ldots, \lambda_n)$ of~\eqref{eq:KStack}, by the formulae
\begin{gather*}
 c_k:= (-1)^{k+1} \sigma_k .
\end{gather*}
The transition matrix between the control basis and the cyclic basis is given by
\begin{gather*}
 \boldsymbol{H}_R=\left[
\begin{matrix}
1&-c_1&\cdots&\cdots&-c_{n-1}\\
0&\ddots&\ddots&\ddots& \vdots\\
0&\ddots&\ddots&\ddots&\vdots \\
\vdots&\ddots&\ddots&\ddots&-c_1 \\
0&\cdots&0&0&1 \\
\end{matrix} \right] ,
\end{gather*}
which can be regarded as a Hankel matrix in a disguised form. We can conclude that the transition matrix between the control basis and the interpolating basis is simply given by the product
\begin{gather}\label{eq:VH}
[\boldsymbol{I}] _{\mathcal{B}_{\rm cont}}^{\mathcal{B}_{\rm int}} =\boldsymbol{V}\boldsymbol{H}_R .
\end{gather}
\end{Remark}
To relate our approach with the classical theory of St\"{ackel} about SoV, based on transformations of coordinates in the conf\/iguration space $\mathcal{Q}$, we need to study which of the Haantjes structures can be projected along the f\/ibers of
$T^*\mathcal{Q}$ by means of the canonical projection map
$\pi\colon T^*\mathcal{Q}\rightarrow \mathcal{Q}$, $(\boldsymbol{q},\boldsymbol{p})\mapsto \boldsymbol{q}$. The following results hold true.
\begin{Proposition} \label{pr:Kp}
The Haantjes operators \eqref{eq:StackHaan} can be projected along the fibers $T^*\mathcal{Q}$ onto the operators
\begin{gather}\label{eq:Kp}
\tilde{\boldsymbol{K}}_{j-1}:=\sum_{r=1}^n\frac{ \tilde{S}_{jr}}{ \tilde{S}_{1r}}
 \frac{\partial}{\partial q_r}\otimes \mathrm{d} q_r, \qquad j=1,\ldots,n-1 ,
\end{gather}
that are still Haantjes operators in the configuration space $\mathcal{Q}$ and are compatible with each other, that is, fulfill the relations \eqref{eq:Lcomp}, \eqref{eq:LtildeC}. Moreover, the Haantjes generator~\eqref{eq:KStack} as well can be projected onto the operator
\begin{gather}\label{eq:Np}
\tilde{\boldsymbol{K}}=\sum_{r=1}^n \lambda_r(\boldsymbol{q}) \frac{\partial}{\partial q_r}\otimes \mathrm{d} q_r, \qquad j=1,\ldots,n-1 .
\end{gather}
Such an operator is a Haantjes operator in $\mathcal{Q}$ and generates the Haantjes operators \eqref{eq:Kp} according to the relations
\begin{gather*}
\tilde{\boldsymbol{K}}_{j-1}=\sum _{r=1}^n \frac{ \tilde{S}_{jr}}{ \tilde{S}_{1r}}
 \pi_r\big(\tilde{\boldsymbol{K}}\big)=\sum_{i =0} ^{n-1} a_i^{(j-1)}(\boldsymbol{q}) \tilde{\boldsymbol{K}}^i,
 \qquad j=1,\ldots,n .
\end{gather*}
\end{Proposition}

\begin{proof}
The components $\tilde{S}_{jr}/\tilde{S}_{1r}$ of
 the Haantjes operators \eqref{eq:StackHaan} in $T^*\mathcal{Q}$, as well as the eigenvalues of $\boldsymbol{K}$,
 in the separation coordinates~$(\boldsymbol{q}; \boldsymbol{p})$
 depend on the coordinates~$\boldsymbol{q}$ only. Therefore, the operators~\eqref{eq:StackHaan} and~\eqref{eq:KStack} can be projected along the f\/ibers of~$T^*\mathcal{Q}$. Moreover, the projected operators inherit the proper\-ties~\eqref{eq:Lcomp},~\eqref{eq:LtildeC} from \eqref{eq:StackHaan}.
\end{proof}

\section{Quasi-bi-Hamiltonian systems}\label{section4}

We derive here the Haantjes structure of a large class of separable systems with $n$ degrees of freedom introduced in~\cite{MT},
that includes a Goldf\/ish system by F.~Calogero and the $L$-systems of Benenti (whose discussion is postponed to Section~\ref{section7}).
Geometrically, such systems can be interpreted as reductions of Gelfand--Zakarevich systems of maximal rank to a~symplectic submanifold of a suitable bi-Hamiltonian manifold~\cite{FP}. To this aim, we have to choose each eigenvalue of the Haantjes generator to be dependent only on the homologous coordinate, so that the generator~\eqref{eq:KStack} becomes the following \emph{Nijenhuis} operator
\begin{gather} \label{eq:NStack}
\boldsymbol{N}:=\sum_{r=1}^n \lambda_r(q_r)\left( \frac{\partial}{\partial q_r}\otimes \mathrm{d} q_r+ \frac{\partial}{\partial p_r}\otimes \mathrm{d} p_r\right) .
 \end{gather}
As before, we assume that its eigenvalues $\lambda_r(q_r)$ are arbitrary smooth functions of their argument, with the restriction that $\lambda_r\neq \lambda_j$ at any point of $\mathcal{Q}$, except possibly for a closed set. Accordingly, we can choose as St\"ackel matrix in equation~\eqref{eq:HStack} the (reverse) Vandermonde matrix of the eigenvalues of~\eqref{eq:NStack}
\begin{gather*}
\boldsymbol{S}(\boldsymbol{q})=\boldsymbol{V}_R=
\begin{bmatrix}
 \lambda_1^{n-1} (q_1)& \lambda_1^{n-2} (q_1)
 &\ldots&1 \\
 \lambda_2^{n-1}(q_2)& \lambda_2^{n-2} (q_2)&\ldots&1 \\
 \ldots & \ldots & \ldots & \ldots \\
 \lambda_n^{n-1} (q_{n})& \lambda_n^{n-2} (q_{n})&\ldots&1
\end{bmatrix}.
\end{gather*}
Computing its inverse, we f\/ind that{\samepage
\begin{gather*}
\big(\boldsymbol{V}_R^{-1}\big)_{jk}=\frac{ (\tilde{V}_R)_{jk}}{ \det(\tilde{V}_R)} =\frac{\partial c_k}{\partial \lambda_j} \frac{1}{\prod\limits_{\stackrel{r=1}{r\neq j}}^n ( \lambda_j- \lambda_r)} .
\end{gather*}
Here the functions $c_k$ are the (opposite of the) coef\/f\/icients of the minimal polynomial of $\boldsymbol{N}$.}

Thus, we have obtained the class of separable Hamiltonian functions
\begin{gather}\label{eq:HFrob}
H_k=\sum_{i=1}^n \frac{\partial c_k}{\partial \lambda_i} \frac{ f_i(q_i,p_i)}{\prod\limits_{\stackrel{j=1}{j\neq i}}^n ( \lambda_i- \lambda_j)},
 \qquad k=1,\ldots,n ,
 \end{gather}
 that has been discussed in \cite{Bl, MTRomp} in the framework of quasi-bi-Hamiltonian (QBH) systems.
 \par
Using equation~\eqref{eq:StackHaan}, and the relations
 \begin{gather*}
 \frac{ \tilde{S}_{jr}}{ \tilde{S}_{1r}} =\frac{ (\tilde{V}_R)_{jr}}{(\tilde{V}_R)_{1r}}=\frac{\partial c_j}{\partial \lambda_r} ,
 \end{gather*}
we f\/ind that such systems admit the simple Haantjes structure $(T^*\mathcal{Q},\omega, \boldsymbol{K}_0=\boldsymbol{I},\boldsymbol{K}_1,\ldots,\boldsymbol{K}_n)$, given by
 \begin{gather}\label{eq:FrobHaan}
 \boldsymbol{K}_{j-1}:=\sum_{r=1}^n \frac{\partial c_j}{\partial \lambda_r} \left( \frac{\partial}{\partial q_r}\otimes \mathrm{d} q_r+ \frac{\partial}{\partial p_r}\otimes \mathrm{d} p_r\right),
 \qquad j=1,\ldots,n-1.
 \end{gather}

 They can be projected onto the Haantjes operators on $\mathcal{Q}$:
 \begin{gather}\label{eq:pFrobHaan}
\tilde{ \boldsymbol{K}}_{j-1}:=\sum_{r=1}^n \frac{\partial c_j}{\partial \lambda_r} \frac{\partial}{\partial q_r}\otimes \mathrm{d} q_r,
 \qquad j=1,\ldots,n-1 .
 \end{gather}

\begin{Proposition} \label{pr:KCont}
The Haantjes operators \eqref{eq:FrobHaan} of a QBH system are the elements of the control basis~\eqref{eq:ContBase} associated to the Nijenhuis operator \eqref{eq:NStack}.
\end{Proposition}
 \begin{proof}
 It is suf\/f\/icient to compute explicitly the transition matrix~\eqref{eq:VH} and to observe that the $i$-th column of such matrix coincides with the eigenvalues of the Haantjes operators~\eqref{eq:FrobHaan}.
 \end{proof}

Due to the fact that the Haantjes operators \eqref{eq:FrobHaan} are generated by the Nijenhuis opera\-tor~\eqref{eq:NStack} through the relations~\eqref{eq:pK}, the Lenard--Haantjes chain formed by the Hamiltonian func\-tions~\eqref{eq:HFrob} is an example of generalized Lenard chain (see, e.g.,~\cite{TT,TGalli12}).

\subsection{A Goldf\/ish system}

In 1996, Calogero studied a solvable system (already introduced by him in 1978) whose Hamiltonian function in canonical coordinates $(q_i,p_i)$ reads
\begin{gather} \label{eq:HGoldC}
H=\sum_{i=1}^n\frac{ e^{a p_i}}{\prod\limits_{\stackrel{j=1}{j\neq i}}^n (q_i-q_j)} .
\end{gather}
The corresponding Newton equations are
\begin{gather*} 
\ddot{q}_k=2\sum_{\stackrel{i=1}{i\neq k}}^n\frac{\dot{ q}_k\dot{q}_i}{ (q_k-q_i)} .
\end{gather*}
This model is the simplest representative of a large class of solvable models called \emph{Goldfish} systems (see~\cite{CalIso} and reference therein).

In the papers \cite{MTPLA, MTRomp}, it was proved that the Goldf\/ish system~\eqref{eq:HGoldC} and the generalized one with Hamiltonian function
\begin{gather} \label{eq:HGoldMT}
H=\sum_{i=1}^n\Bigg(\frac{ e^{a p_i}}{\prod\limits_{\stackrel{j=1}{j\neq i}}^n (q_i-q_j)}+b q_i\Bigg),\qquad b\in \mathbb{R} ,
\end{gather}
and with Newton equations
\begin{gather*} 
\ddot{q}_k=2\sum_{\stackrel{i=1}{i\neq k}}^n\frac{\dot{q}_k\dot{q}_i}{ (q_k-q_i)}-a b \dot{q}_k ,
\end{gather*}
admit a quasi-bi-Hamiltonian structure which ensures the separability of the associated H-J equation directly in the symplectic coordinates $(\boldsymbol{q};\boldsymbol{p})$. We wish to point out that the gene\-ra\-lized system~\eqref{eq:HGoldMT} admits also the~$\omega\mathcal{H}$ structure~\eqref{eq:FrobHaan}, shared by all quasi-bi-Hamiltonian systems. Therefore, in turn, it belongs to the generalized St\"ackel class~\eqref{eq:HStack}.

The Hamiltonian function of the Goldf\/ish system~\eqref{eq:HGoldMT} arises from equation~\eqref{eq:HFrob} with $\lambda_i\equiv q_i$, $i=1,\ldots,n$,
 with the following choice of the St\"ackel functions
 \begin{gather*}
f_i:=e^{ap_i}+b q_i^n ,
\end{gather*}
taking also into account the Jacobi identity~\cite{AKN}
\begin{gather*}
\sum_{i=1}^n\frac{q_i^n}{\prod\limits_{\stackrel{j=1}{j\neq i}}^n (q_i-q_j)}=
\sum_{i=1}^n q_i .
\end{gather*}

\section{Classical systems of St\"ackel}\label{section5}

The classical separable St\"ackel systems arise from equation~\eqref{eq:HStack}
by choosing as St\"ackel functions the homogeneous quadratic functions in the momenta
\begin{gather}\label{eq:cSf}
f_k:=\frac{1}{2} p_k^2+W_k(q_k) .
\end{gather}
The functions $W_k(q_k)$ are components of the so-called St\"ackel multiplicator \cite{Pars}.
With this choice, the Hamiltonian function \eqref{eq:HStack} takes the form
\begin{gather*}
H=\frac{1}{2} \sum_{j=1}^n g^j(\boldsymbol{q}) p_j^2+ V(\boldsymbol{q}),
\end{gather*}
where the functions
\begin{gather*} 
g^j(\boldsymbol{q})=\frac{ \tilde{S}_{1j}}{\det(\boldsymbol{S})}
\end{gather*}
can be interpreted as the diagonal components $g^j:=g^{jj}$ of the inverse of a metric tensor~$g$ over the conf\/iguration space~$\mathcal{Q}$:
\begin{gather} \label{eq:Sg}
\boldsymbol{G}:=\sum_{j=1}^n g^j \frac{ \partial}{\partial q_j} \otimes \frac{ \partial}{\partial q_j}.
\end{gather}
Also,
\begin{gather*} 
V(\boldsymbol{q})=\sum_{j=1}^n g^j W_j
\end{gather*}
is the potential energy.
The presence of the metric~\eqref{eq:Sg} allows one to construct the contravariant form of the Haantjes operators~\eqref{eq:Kp}, which are still diagonal, with components
\begin{gather*}
(K_{j-1})^{ii}=(K_{j-1})^i_i g^{i}=\frac{ \tilde{S}_{j i}}{\det (\boldsymbol{S})}, \qquad i, j=1,\ldots,n .
\end{gather*}

\begin{Proposition} \label{pr:KKill}
The tensor fields~\eqref{eq:Kp} are Killing tensors for the metric~\eqref{eq:Sg} and are in involution with respect to the Schouten bracket of two symmetric contravariant tensors.
 \end{Proposition}

\begin{proof}
The result follows immediately from the f\/irst involution relation~\eqref{eq:invTV}.
\end{proof}

\section{The Jacobi--Calogero model}\label{section6}

In this section, we analyze in the framework of Haantjes geometry the celebrated rational Calogero model describing particles on a line, interacting with an inverse square potential~\cite{CalScp}. We shall limit ourselves to the case of three particles, already introduced by Jacobi in~1866~\cite{Jac}, but totally forgotten till his contribution was re-discovered in~\cite{Pere}. The Jacobi--Calogero Hamiltonian function reads
\begin{gather} \label{eq:HCal}
H=\frac{1}{2} \big(p_x^2 +p_y^2+p_z^2\big)+V_{\rm Cal} ,
\end{gather}
where the
potential energy $V_{\rm Cal}$ is
\begin{gather*} \label{eq:VCal}
V_{\rm Cal}=\frac{a}{(x-y)^2}+ \frac{a}{(y-z)^2}+\frac{a}{(z-x)^2}, \qquad a\in \mathbb{R} ,
\end{gather*}
and $(x,y,z)$ are the coordinates of the three particles on the line.
The conf\/iguration space of this system is $\mathcal{Q}=\mathcal{E}^3 \setminus \Delta$, where
$\mathcal{E}^3$ denotes the 3-dimensional Euclidean af\/f\/ine space and $\Delta:=\{x=y, x=z,y=z \}$ the set of the collision planes. Its cotangent bundle
$T^*\mathcal{Q}\simeq\mathcal{Q}\times \mathcal{E}^3 $ is the phase space of the model.

Two bi-Hamiltonian structures has been worked out for this system. However, in the f\/irst one~\cite{MMars}, the computation of a second Poisson operator $\boldsymbol{P}_1$ and therefore of a Nijenhuis operator $\boldsymbol{N}:=\boldsymbol{P}_1\boldsymbol{\Omega}$ seems to be prohibitively complicated and has not been carried out explicitly. The other bi-Hamiltonian structure \cite{Tsi} is an ``irregular" one and does not provide integrals of motion dif\/ferent from the Hamiltonian~\eqref{eq:HCal}.

Here, we will compute f\/ive Haantjes operators that turn out to be very simple, that is, at most quadratic in the coordinates and momenta, and provide integrals of motion by means of three independent Lenard--Haantjes chains.

As is well known, the model is \emph{maximally superintegrable}, namely, it admits f\/ive independent integrals of motion in involution \cite{MPW, TWR}. Furthermore, it is also \emph{multi-separable}. In fact, in the interesting paper \cite{BCRcal}, it has been proved that, besides the known circular cylindrical coordinates \cite{Cal69}, there are four other webs in which the associate HJ equation is separable: spherical, parabolic, oblate spheroidal and prolate spheroidal. All such webs have a common axis of rotational symmetry. Using the f\/irst three of them, we will be able to construct the Haantjes structures of the model. The other two webs, oblate and prolate spheroidal, do not provide further independent Haantjes structures. Specif\/ically, we will write down the Calogero Hamiltonian function~\eqref{eq:HCal} (the source) and two integrals of motion (the $\boldsymbol{K}$-images) in each of the separable webs and we will apply Theorem~59 of~\cite{TT2014}, that assures the construction of a~Haantjes structure for which the separable coordinates are DH coordinates.

Due to the absence of external force f\/ields, the linear momentum
\begin{gather*}
p_x+p_y+p_z
\end{gather*}
is a (linear) integral of motion. In the conf\/iguration space $\mathcal{Q}$, it is equivalent to the conserved scalar quantity $(\vec{p}\cdot \vec{u})$, where~$\vec{p}$ and $\vec{u}$ are the vectors
\begin{gather*}
\vec{p} := p_x \vec{e}_x+ p_y \vec{e}_y+\vec{e}_z, \qquad \vec{u}:=\vec{e}_x+ \vec{e}_y+\vec{e}_z,
\end{gather*}
and \looseness=-1 $(\vec{e}_x, \vec{e}_y,\vec{e}_z)$ is a basis of three orthonormal vectors in~$\mathbb{E}_3$.
This fact amounts to the rotational symmetry of the model around the axis $(O,\vec{u})$, that is, the straight line passing through the origin of the coordinates~$O$ and parallel to the vector~$\vec{u}$. Thus, such an axis is a symmetry axis for each of the separable webs above-mentioned. Following~\cite{BCRcal}, we consider the integral of motion
\begin{gather*}
 H_2 =\frac{1}{6} \big(\vec{L}_0\cdot \vec{u}\big)^2+
\left|\vec{r}\times \frac{\vec{u}}{|\vec{u}|}\right|^2V_{\rm Cal}
= \frac{1}{6} \big( (yp_z-zp_y)+(zp_x-xp_z)+(xp_y-yp_x)\big)^2\\
\hphantom{H_2 =}{}
+\frac{1}{3}\big((x-y)^2
+(x-z)^2+(y-z)^2\big) V_{\rm Cal} ,
\end{gather*}
related to the axial angular momentum. This integral has a privileged role for it is separable in each of the separated webs above-mentioned. Thus, writing down the integral and the Hamiltonian function~\eqref{eq:HCal} in one of the separated webs, and
using equation~\eqref{eq:LSoV} we obtain a diagonal operator (consequently a Haantjes one) which in cartesian coordinates reads
\begin{gather*}
\boldsymbol{K}_1=
\left[
\begin{array}{@{}c|c@{}}
\boldsymbol{A} &0_3 \\
 \hline
 \boldsymbol{B}&\boldsymbol{A}\\
 \end{array}
\right ],
\end{gather*}
where
\begin{gather*}
\boldsymbol{A}=\frac{1}{3}
\left[
\begin{matrix}
(y-z)^2&(y-z)(z-x)&(y-z)(x-y)\\
(y-z)(z-x)&(x-z)^2&(z-x)(x-y)\\
(y-z)(x-y)&(z-x)(x-y)&(x-y)^2
\end{matrix}
\right ] ,
\\
\boldsymbol{B}=\frac{1}{3}\big ( (x-y)p_z+(y-z)p_x+(z-x)p_y \big)
\left[
\begin{matrix}
0&1&-1\\
-1&0&1\\
1&-1&0
\end{matrix}
\right ] .
\end{gather*}
Such an operator provides
\begin{gather*}
\mathrm{d} H_2 =\boldsymbol{K}_1^T\mathrm{d} H ,
\end{gather*}
which is the f\/irst element common to the three Lenard--Haantjes chains presented in equations~\eqref{eq:LHcJC}.

Now, we shall focus on the separable webs.
\subsection{Cylindrical Haantjes operator}

Let us consider the integral of the (square of the) linear momentum
\begin{gather*}
H_{\rm cil}=\frac{1}{2}(\vec{p}\cdot \vec{u})^2=\frac{1}{2}(p_x +p_y+p_z)^2 .
\end{gather*}
Once we write it in cylindrical circular coordinates with axes $(O,\vec{u})$ together with the Hamiltonian function \eqref{eq:HCal}, according to equation~\eqref{eq:LSoV}
we can def\/ine a second uniform Haantjes operator, given in cartesian coordinates by
\begin{gather*}
\boldsymbol{K}_{\rm cil}=
\left[
\begin{array}{@{}c|c@{}}
\boldsymbol{A}_{\rm cil} &\boldsymbol{0}_3 \\
 \hline
 \boldsymbol{0}_3&\boldsymbol{A}_{\rm cil}\\
 \end{array}
\right ] ,
\qquad \text{where} \qquad
\boldsymbol{A}_{\rm cil}=
\left[
\begin{matrix}
1&1&1\\
1&1&1\\
1&1&1
\end{matrix}
\right ] .
\end{gather*}

\subsection{Spherical Haantjes operator}
Analogously, we consider the following integral of motion
\begin{gather*}
 H_{\rm sph} = \frac{1}{2} |\vec{L}_O|^2 +|\vec{r}|^2 V_{\rm Cal} \\
\hphantom{H_{\rm sph}}{} = \frac{1}{2}\big ( (yp_z-zp_y)^2+(zp_x-xp_z)^2+(xp_y-yp_x)^2\big)+\big(x^2 +y^2+z^2\big) V_{\rm Cal}
 \end{gather*}
 related to the (square) module of the angular momentum. From the expression of this integral in spherical coordinates and from the Hamiltonian function \eqref{eq:HCal}, we construct a Haantjes operator that in cartesian coordinates reads
\begin{gather*}
\boldsymbol{K}_{\rm sph}=
\left[
\begin{array}{@{}c|c@{}}
\boldsymbol{A}_{\rm sph} &0_3 \\
 \hline
 \boldsymbol{B}_{\rm sph}&\boldsymbol{A}_{\rm sph}\\
 \end{array}
\right ],
\end{gather*}
where
\begin{gather*}
\boldsymbol{A}_{\rm sph}=
\left[
\begin{matrix}
y^2+z^2&-xy&-zx\\
-xy&x^2+z^2&-yz\\
-zx&-yz&x^2
+y^2
\end{matrix}
\right ] ,
\\
\boldsymbol{B}_{\rm sph}=
\left[
\begin{matrix}
0&yp_x-xp_y&zp_x-xp_z\\
-(yp_x-xp_y)&0&zp_y-yp_z\\
-(zp_x-xp_z)&-(zp_y-yp_z)&0
\end{matrix}
\right ] .
\end{gather*}

\subsection{Parabolic Haantjes operator}
We consider the following integral of motion
\begin{gather*}
 H_{\rm par} = \frac{1}{2} \big ( (\vec{p}\cdot \vec{u}) (\vec{p}\cdot \vec{r}) -(\vec{r}\cdot\vec{u})(\vec{p}\cdot\vec{p})\big)-(\vec{r}\cdot\vec{u})V_{\rm Cal} \\
\hphantom{H_{\rm par}}{}
 = \frac{1}{2}\big ( (p_x+p_y+p_z)(xp_x+yp_y+zp_z)-(x+y+z)\big(p_x^2+p_y^2+p_z^2\big)\big)-(x+y+z)V_{\rm Cal}
\end{gather*}
related to the product of the axial with the radial linear momentum. The associated Haantjes operator is
\begin{gather*}
\boldsymbol{K}_{\rm par}=
\left[
\begin{array}{@{}c|c@{}}
\boldsymbol{A}_{\rm par} &0_3 \\
 \hline
 \boldsymbol{B}_{\rm par}&\boldsymbol{A}_{\rm par}\\
 \end{array}
\right ],
\end{gather*}
where
\begin{gather*}
\boldsymbol{A}_{\rm par}=\frac{1}{2}
\left[
\begin{matrix}
-2(y+z)&(x+y)&(x+z)\\
(x+y)&-2(x+z)&(y+z)\\
(x+z)&(y+z)&-2(x+y)
\end{matrix}
\right ] ,
\\
\boldsymbol{B}_{\rm par}=\frac{1}{2}
\left[
\begin{matrix}
0&p_y-p_x&p_z-p_x\\
-(p_y-p_x)&0&p_z-p_y\\
-(p_z-p_x)&-(p_z-p_y)&0
\end{matrix}
\right ] .
\end{gather*}
The following result holds.

\begin{Proposition}
The three Haantjes structures $(T^*\mathcal{Q}, \omega, \boldsymbol{K}_0=\boldsymbol{I}_6,\boldsymbol{K}_1,\boldsymbol{K}_{\rm cyl})$,
$(T^*\mathcal{Q}, \omega, \boldsymbol{K}_0=\boldsymbol{I}_6,\boldsymbol{K}_1,\boldsymbol{K}_{\rm sph})$,
$(T^*\mathcal{Q}, \omega, \boldsymbol{K}_0=\boldsymbol{I}_6,\boldsymbol{K}_1,\boldsymbol{K}_{\rm par})$ together with the Hamiltonian function~\eqref{eq:HCal} ge\-ne\-rate three Lenard--Haantjes chains with two common elements
\begin{gather} \label{eq:LHcJC}
\begin{split} &
\xymatrix{
&&&\boldsymbol{K}_{\rm cyl}^T \mathrm{d} H=\mathrm{d} H_{\rm cyl}, \\
\boldsymbol{K}_0^T\mathrm{d} H=\mathrm{d} H_1 ,& \boldsymbol{K}_1^T\mathrm{d} H=\mathrm{d} H_2 ,&
\ar[ur] \ar[r] \ar[dr] &\boldsymbol{K}_{\rm sph}^T \mathrm{d} H=\mathrm{d} H_{\rm sph}, \\
&&&\boldsymbol{K}_{\rm par}^T \mathrm{d} H=\mathrm{d} H_{\rm par}.
}
\end{split}
\end{gather}

\end{Proposition}
\begin{proof}
In any of the separable webs above-mentioned the operators $ \boldsymbol{K_0}$ and $ \boldsymbol{K}_1$ take a diagonal form. Furthermore, $ \boldsymbol{K}_{\rm cyl}$, $ \boldsymbol{K}_{\rm sph}$, $ \boldsymbol{K}_{\rm par}$, by construction, are diagonal in the cylindrical, spherical, parabolic webs, respectively. Then, they fulf\/ill all the conditions of Def\/inition~\ref{def:HM}.
\end{proof}

\begin{Remark}
The existence of more than one independent Lenard--Haantjes chain is due to the superintegrability of the Calogero model. However, only two of the three chains are independent as
\begin{gather*}
\mathrm{d} H_1 \wedge\mathrm{d} H_2\wedge\mathrm{d} H_{\rm cyl} \wedge \mathrm{d} H_{\rm sph}\wedge \mathrm{d} H_{\rm par}=0 ,
\end{gather*}
and
\begin{gather*}
\mathrm{d} H_1 \wedge\mathrm{d} H_2\wedge\mathrm{d} H_{\rm cyl} \wedge \mathrm{d} H_{\rm sph} \neq 0,\\
\mathrm{d} H_1 \wedge\mathrm{d} H_2\wedge\mathrm{d} H_{\rm cyl} \wedge \mathrm{d} H_{\rm par} \neq 0,\\
\mathrm{d} H_1 \wedge\mathrm{d} H_2\wedge \mathrm{d} H_{\rm sph}\wedge \mathrm{d} H_{\rm par} \neq 0.
\end{gather*}

Therefore, an additional independent integral is required in order to prove the maximal superintegrability of the model. The additional integral is the cubic one in the momenta
\begin{gather*}
 H_3:=\frac{1}{3}\big(p_x^3 +p_y^3+p_z^3\big)+a \left(\frac{p_x+p_y}{(x-y)^2}+\frac{p_x+p_z}{(x-z)^2}+ \frac{p_y+p_z}{(y-z)^2}+\frac{p_x+p_z}{(x-z)^2}\right) ,
\end{gather*}
which is in involution both with $H_1$ and $H_2$. The problem of f\/inding a~Haantjes structure that involves such an integral is under investigation.
\end{Remark}

\begin{Remark}
According to Proposition~\ref{pr:Kp}, all the previous Haantjes operators can be projected
onto the conf\/iguration space. Each projection is simply given by the f\/irst block of the representative matrix, that is, by
$\boldsymbol{I}_3$, $\boldsymbol{A}$, $\boldsymbol{A}_{\rm cil}$, $\boldsymbol{A}_{\rm sph}$, $\boldsymbol{A}_{\rm par}$.
These two-tensors in the conf\/iguration space coincide with the mixed form $(1,1)$ of the Killing tensors found in~\cite{BCRcal}.
\end{Remark}

\section{A Haantjes route to Benenti systems} \label{section7}

In this section, we will prove that the $\boldsymbol{L}$-systems introduced by S.~Benenti~\cite{Ben92, Ben93}
 and discussed in~\cite{IMM} within a bi-Hamiltonian framework, can be recovered by projecting onto the conf\/iguration space the Haantjes operators of the QBH systems~\eqref{eq:FrobHaan}. This can be done by choosing the classical quadratic functions in the momenta~\eqref{eq:cSf} as St\"ackel functions in the Hamiltonian~\eqref{eq:HFrob}. Indeed, the components of the metric~\eqref{eq:Sg} turn out to be
\begin{gather} \label{eq:gQBH}
g^{i}=\frac{ (\tilde{V}_R)_{1i}}{\det(\tilde{V}_R)} =\frac{\partial c_1}{\partial \lambda_i} \frac{1}{\prod\limits_{\stackrel{j=1}{j\neq k}}^n ( \lambda_i- \lambda_j)} =
\frac{1}{\prod\limits_{\stackrel{j=1}{j\neq i}}^n ( \lambda_i- \lambda_j)} .
\end{gather}
Now we are able to prove the following
\begin{Proposition}
The projected Haantjes operators \eqref{eq:pFrobHaan} are Killing tensors w.r.t.\ the metric~\eqref{eq:gQBH} and commute with each other. As for the QBH systems $\lambda_r(\boldsymbol{q})=\lambda_r(q_r)$, the projected operator~\eqref{eq:Np}, which we shall denote by~$\boldsymbol{L}$, is a~Nijenhuis operator and generates the Killing tensors~\eqref{eq:pFrobHaan} by means of the relations
\begin{gather}
 \tilde{\boldsymbol{K}_0} = \boldsymbol{L}^0=\boldsymbol{I}, \nonumber\\
\tilde{\boldsymbol{K}_\alpha} = -\sum_{j=0}^{\alpha -1} c_{\alpha-j} \boldsymbol{L}^j+\boldsymbol{L}^\alpha , \qquad \alpha=1,\ldots, n-1.
\label{eq:BenTens}
\end{gather}
Moreover, it is a $L$-tensor or a conformal Killing tensor of trace-type, i.e., it fulfills the relation
\begin{gather} \label{eq:NCKT}
[\boldsymbol{L},\boldsymbol{G}]=-2 \boldsymbol{X} \odot \boldsymbol{G}, \qquad
\boldsymbol{X}=\boldsymbol{G} \mathrm{d} (\operatorname{tr}(\boldsymbol{L}) ) .
\end{gather}
Here, the symbols $[\,,\,]$ and $\odot$ denote the Schouten bracket and the symmetric product of two contravariant tensor fields, respectively.
 Furthermore, the potential functions~\eqref{eq:StackTV} form the Haantjes chain in~$\mathcal{Q}$
\begin{gather} \label{eq:Hcp}
\tilde{\boldsymbol{K}}_{j-1}^T\mathrm{d} V_1= \mathrm{d} V_{j},\qquad j=1,\ldots, n .
\end{gather}
\end{Proposition}

\begin{proof}
The f\/irst assertion is a direct consequence of Proposition~\ref{pr:KKill} and of the compatibility condition~\eqref{eq:Lcomp}.
The generating formula~\eqref{eq:BenTens} follows from Proposition~\ref{pr:KCont}.
Property~\eqref{eq:NCKT} is a~consequence of equation~\eqref{eq:BenTens}, which for $\alpha=1$ and in contravariant form implies
\begin{gather*}
\boldsymbol{L}=c_1 \boldsymbol{G}+\tilde{\boldsymbol{K}_1},
\end{gather*}
and of the properties of the Schouten bracket. Indeed,
\begin{gather*}
[\boldsymbol{L},\boldsymbol{G}]= [c_1 \boldsymbol{G}+\tilde{\boldsymbol{K}_1},\boldsymbol{G}]=-2 [c_1 ,\boldsymbol{G}] \odot \boldsymbol{G} .
\end{gather*}
Finally, the potential functions $V_j(\boldsymbol{q})$ in equation~\eqref{eq:StackTV} can be projected naturally along the f\/ibers of $T^{*} \mathcal{Q}$. Therefore, the second Haantjes chain of~\eqref{eq:HcTV} can also be projected onto $\mathcal{Q}$, giving equation~\eqref{eq:Hcp}.
\end{proof}

\subsection*{Acknowledgements}

The work of P.T.~has been partly supported by the research project FIS2015-63966, MINECO, Spain
and partly by ICMAT Severo Ochoa project SEV-2015-0554 (MINECO).
G.T.~acknow\-led\-ges the f\/inancial support of the research project PRIN 2010-11 ``Geometric and analytic theory of Hamiltonian systems in f\/inite and inf\/inite dimensions''. Moreover, he thanks G.~Rastelli for interesting discussions about the Jacobi--Calogero model.
We also thank the anonymous referees for a careful reading of the manuscript and for several useful suggestions.

\pdfbookmark[1]{References}{ref}
\LastPageEnding

\end{document}